\documentclass [11pt]{article}
\usepackage{amsmath,amsthm,amsfonts,amscd,eucal,latexsym,amssymb,mathrsfs}
\usepackage{epsfig}  
\input xy
\xyoption{all}
 
\def\bR{{\mathbb R}}

\def\beq{\begin{eqnarray}}
\def\eeq{\end{eqnarray}}

\def\ga{\gamma}

\newcommand{\supp}[1]{\text{supp }{#1}}

\def\vsp{\vspace{0.2cm}}

\def\sse #1 {\vsp\ifhmode{\par}\fi\refstepcounter{subsection}
  \noindent {\bf\thesubsection}. {\em #1}.\quad
  \addcontentsline{toc}{subsection}{\protect\numberline{\thesubsection} #1}%
  }

\def\ssb #1 {\vsp\ifhmode{\par}\fi\refstepcounter{subsection}
  \noindent {\bf\thesubsection.} {\bf #1.}\quad
  \addcontentsline{toc}{subsection}{\protect\numberline{\thesubsection} #1}%
  }

\def\ssa #1 {\ifhmode{\par}\fi\refstepcounter{subsection}
  \noindent {\bf\thesubsection.} {\bf #1.}\quad
  \addcontentsline{toc}{subsection}{\protect\numberline{\thesubsection} #1}%
  }

\def\remark {\vsp\ifhmode{\par}\fi\noindent\noindent {\bf Remark:} 
}

\newtheorem{proposition}{Proposition}[section]

\newtheorem{definition}{Definition}[section]

\begin{document} 
%
%
 
\par 
\bigskip 
\LARGE 
\noindent 
{\bf Influence of quantum matter  fluctuations on geodesic deviation} 
\bigskip 
\par 
\rm 
\normalsize 
 

\large
\noindent 
{\bf Nicol\`o Drago$^{a}$}, {\bf Nicola Pinamonti$^{b}$} \\
\par
\small
\noindent Dipartimento di Matematica, Universit\`a di Genova,  
via Dodecaneso, 35 I-16146 Genova,  Italy. \smallskip
\smallskip

\noindent E-mail: 
$^a$drago@dima.unige.it,   $^b$pinamont@dima.unige.it\\ 

\normalsize

\par 
 
\rm\normalsize 

\rm\normalsize 
 
 
\par 
\bigskip 
 
\rm\normalsize 

\par 
\bigskip

\rm\normalsize 

\noindent 
\small 
{\bf Abstract}.
We study the passive influence of quantum matter fluctuations on the expansion parameter of a congruence of timelike geodesics in a semiclassical regime.
In particular, we show that, the perturbations of this parameter can be considered to be elements of the algebra of observables of matter fields at all perturbative orders.
Hence, once a quantum state for matter is chosen, it is possible to explicitly evaluate the behavior of geometric fluctuations.
After introducing the formalism necessary to treat similar problems, in the last part of the paper, we estimate the approximated probability of having a geodesic collapse in a flat spacetime due to those fluctuations. 
\normalsize

\vskip .3cm

\section{Introduction}
Usually, semiclassical Einstein equations are used to evaluate the influence of quantum matter on classical curvature. 
These equations are obtained discussing quantum field theories propagating on classical curved backgrounds $(M,g)$ and equating the expectation value of the matter stress tensor $T_{ab}$, in a suitable state $\omega$, with the classical Einstein tensor \cite{Wald}. Assuming units $c=\hbar=8\pi G_N=1$ they read
\begin{equation}\label{eq:semiclassical-einstein-equation}
G_{ab} = \langle T_{ab} \rangle_\omega.
\end{equation}
We recall that one of the most promising alternatives to the inflaton driven inflationary cosmological model are based on this equation \cite{Starobinsky, kls}. 
Here we want to stress that, a similar equation can make sense only in the regime where the fluctuations of $T$ are negligible. Unfortunately, quantum fluctuations of point-like fields like $T_{ab}(x)$ are divergent and thus they can never be neglected.
These problems becomes milder when smearing of $T_{ab}$ is taken into account and, in particular, when the smearing function is sufficiently large these fluctuations become very small.
Unfortunately, the smeared version of \eqref{eq:semiclassical-einstein-equation} is a non local equation which furthermore, breaks 
covariance.

There is however an alternative way to deal with the semiclassical regime  
taking also into account the influence of fluctuations of $T_{ab}$ on $G_{ab}$. We can actually impose that the metric is described by a random field and we can assume that quantum fluctuations of the curvature tensor are passively sourced by $T_{ab}$.
This point of view is the one proper of the theory of stochastic gravity. A nice review on this subject can be found in \cite{Stochastic}.
Following similar ideas, 
in this paper, we study the influence of quantum matter fluctuations on the deviation equation for congruences of timelike geodesics. 

More precisely, we study the passive influence of the fluctuations of $T_{ab}$ on the geodesic deviation equation. 
In the physical literature, similar problems have been discussed in various cases.
For example, Carlip, Mosna and Pitelli in \cite{Carlip}
analyzed the case of two-dimensional dilation gravity models,
employing the construction of probability distributions of the expectation values of 
$T_{ab}$ provided by Fewster, Ford and Roman \cite{FewsterFordRoman}.
In that paper, the authors show that, the quantum nature of the stress tensor induce a sharp focusing of light cones at Planck scale. Furthermore, they estimate the mean time in which the focusing occurs.
In \cite{Carlip}, the complicated interaction equation which describe the behavior of the geodesic congruence expansion parameter is solved dividing the time in finite intervals of the order of the Planck scale, and treating $T_{ab}$ as a constant on every interval; the value of this constant is chosen according to the probability distribution of $T_{ab}$.
The probability distribution for focusing is then obtained counting the occurrences of the light focusing events in a large number of random sequences.
Another recent work employing similar ideas is \cite{ParkinsonFord}, where the authors discuss the first order influence of matter fluctuations on null geodesics.

In this paper, we use a quantum massless minimally coupled scalar field $\varphi$  as a model for matter, and, after quantizing it on a fixed background which solves the semiclassical Einstein equation \eqref{eq:semiclassical-einstein-equation}, we study the evolution equation for the expansion parameter on a congruence of timelike geodesics. We assume that the influence of the stress tensor fluctuations are only manifest in the expansion parameter, hence its effect on the shear and on the twist of the congruence is neglected. Then we shall show that it is possible to solve the Raychaudhuri equation for the expansion parameter within the algebra of quantum matter fields restricted on the congruence of geodesics.

The analysis we shall perform makes use of methods developed in the so called
perturbative algebraic quantum field theory \cite{BF09,FredenhagenRejzner}.
The corresponding solution is thus given in terms of a formal power series in the coupling parameter.
However, at every order, every element of the series is a well defined field in the algebra of matter fields.
In order to obtain this result, there is the need of employing renormalization and regularization techniques proper of interacting field theory 
in analogy with the discussion presented in \cite{BF00}.
There are thus certain renormalization constants that need to be fixed.

Once a state is chosen it is possible to estimate the correlation functions for the field $\phi$ whose logarithmic derivative with respect to the proper time along the geodesic gives the expansion parameter $\theta$.
We shall accomplish this task for the case of a four dimensional Minkowski spacetime as background considering a congruence of timelike geodesics with vanishing shear and twist.
The perturbative solution is explicitly computed up to the second order in the Newton constant.
Furthermore the probability distribution for the field $\phi$ is approximated by a Gaussian distribution. Under this approximation, the computation of the mean time to collapse confirms the results obtained for two-dimensional dilatonic gravity in \cite{Carlip}. 

The paper is organized as follows: in the next section we recall the basic notion on quantization of the massless scalar field $\varphi$ on a globally hyperbolic spacetime in the context of the Algebraic Quantum Field Theory, and introduce the $\ast$-algebra of interest.
Section \ref{sec:3} is devoted to the construction of the algebra associated to the geometrical field $\phi$, introduced using the Raychaudhuri equation.
Finally, in Section \ref{sec:4} we apply the previously introduced formalism to the flat case, obtaining an expression for $\phi$ at second order in the coupling constant and approximating its probability distribution to a Gaussian. Finally some conclusions are given in the last section.

\section{Algebraic formulation of quantum field theory on curved spacetime}\label{sec:2}

Let $(M,g)$ be a globally hyperbolic spacetime, where $M$ is a smooth manifold and 
$g$ is a Lorentzian metric. We shall consider the
simplest possible matter field propagating on such background, 
namely, the massless minimally coupled field $\varphi$ whose classical equation of motion is
\begin{equation}\label{eq:equation-of-motion}
-\Box_g \varphi = 0 \;.
\end{equation}
Above $\Box_g$ is the D'Alembert operator constructed out of the Lorentzian metric $g$.
The quantization of this system is very well under control. We shall treat it according to the principle of algebraic quantization \cite{BFV,HW01,HW02}. Following \cite{BF09,BDF}, we introduce the non commutative $*$-algebra $\mathcal{A}(M)$ whose elements are the smooth functionals $F:\mathcal{E}(M)\rightarrow\mathbb{C}$ over field configurations, which are differentiable with compact support and 
whose functional derivative $F^{(n)}(\varphi)$ have vanishing wavefront set.
Here differentiability means that for each pair of smooth functions $h,\varphi\in\mathcal{E}(M)$ the function $\lambda\rightarrow F(\varphi+\lambda h)$ is often infinitely differentiable and the $n$-th derivative at $\lambda=0$ defines a symmetric distribution with compact support $F^{(n)}(\varphi)$ on $\mathcal{D}(M^n)$ with the property
\begin{equation*}
\frac{d^n}{d\lambda^n}_{|\lambda=0}F(\varphi+\lambda h)
=\left\langle F^{(n)}(\varphi),h^{\otimes n}\right\rangle,
\end{equation*}
and such that the application $(\varphi,h)\rightarrow\langle F^{(n)}(\varphi),h^{\otimes n}\rangle$ is continuous in the proper topology.
The support of $F$ is the closure of the union of the support of $F^{(1)}(\varphi)$ for all $\varphi\in\mathcal{E}(M)$. 

The product on $\mathcal{A}(M)$, denoted by $\star$,  and defined by
\begin{equation}\label{eq:prodotto *}
(F\star G)(\varphi)
=F(\varphi)G(\varphi)
+\sum_{n=0}^{+\infty}\frac{\imath^n}{2^nn!}\Delta^{\otimes n}(F^{(n)}(\varphi),G^{(n)}(\varphi)),
\end{equation}
is non commutative.
As an immediate consequence of the definition, the algebra $\mathcal{A}(M)$ is nothing but the algebra generated by the functionals $\boldsymbol{\varphi}(f):=\int dx f(x)\varphi(x),\ f\in\mathcal{D}(M)$, which satisfy
\begin{equation}\label{eq:relazioni algebriche}
[\boldsymbol{\varphi}(f),\boldsymbol{\varphi}(h)]_\star = i\Delta(f,h),
\end{equation}
where $\Delta$ is the advanced-minus-retarded solution of $-\Box$, whose existence is assured by the global hyperbolicity of $M$.
Moreover the $\ast$-operation on $\boldsymbol{\varphi}(f)$ acts as follows
\begin{equation}\label{operatore*}
\boldsymbol{\varphi}(f)^* = \boldsymbol{\varphi}(\overline{f}).
\end{equation}

States on $\mathcal{A}(M)$ are defined as linear functionals $\omega:\mathcal{A}(M)\rightarrow\mathbb{C}$ which are positive, i.e. $\omega(A^*A)\geq 0\ \forall A\in\mathcal{A}(M)$, and normalized ($\omega(\mathbf{1})=1, \mathbf{1}$ identity of the algebra).
These conditions imply that a state is completely described by his $n-$point functions
\begin{equation*}
\omega(f_1,\ldots,f_n)
=\omega(\boldsymbol{\varphi}(f_1)\star\ldots\star\boldsymbol{\varphi}(f_n)),
\end{equation*}
whose totally antisymmetric part is fixed by the second relation in (\ref{eq:relazioni algebriche}).

In the next step we want to include fields like the stress tensor in the algebra. This is necessary in order to evaluate the influence of quantum matter on congruences of timelike geodesics.
Notice that the classical stress tensor of the analyzed theory is
\begin{equation}
T_{\mu\nu} = \partial_\mu\varphi\partial_\nu\varphi - \frac{1}{2} g_{\mu\nu} \partial^\lambda\varphi\partial_\lambda\varphi\;.
\end{equation}
Unfortunately the stress tensor involves the product of point-like fields and, for this reason, it cannot be constructed as an element of $\mathcal{A}(M)$.  

However, as discussed in \cite{BF09,BDF}, after deforming the product of $\mathcal{A}(M)$ in a suitable way, it is possible to easily extend it to non linear local fields like $T_{\mu\nu}$. 
The aforementioned deformation is realized by means of an Hadamard function $H$ in such a way that the deformed product on generators reads
$$
\boldsymbol{\varphi}(f) \star_H \boldsymbol{\varphi}(h) = \boldsymbol{\varphi}(f)\boldsymbol{\varphi}(h)+H(f,h)\;.
$$
For generic elements $F,G$ the $\star_H$ product is defined as in formula \eqref{eq:prodotto *} where $\imath\Delta/2$ is substituted by $H$. We recall that the Hadamard function $H$ is actually an element of $\mathcal{D}'(M^2)$, whose antisymmetric part is $\imath\Delta(x,y)/2$, furthermore, it weakly satisfies the equation of motion up to smooth functions and it enjoys the micro local spectrum condition \cite{Radzikowski, BFK}. Notice also that the new product $\star_H$ preserves the relation (\ref{eq:relazioni algebriche}).

We shall indicate the extended algebra by $\mathcal{F}(M)$; its elements are again the functionals over smooth field configuration ($F:\mathcal{E}(M)\rightarrow\mathbb{C}$) which satisfy certain property on the wave front set of the functional derivatives, namely
\begin{equation}\label{eq:spetral-condition}
\mathcal{F}(M) = \left\{ F \text{ differentiable with compact support}\;,\quad WF(F^{(n)}(\varphi)) \cap (\overline{V_+}^n\cup \overline{V_-}^n) = \emptyset   \right\}.
\end{equation}

As discussed in \cite{Radzikowski}, the set of conditions expressed above characterize the Hadamard function $H$ up to smooth functions, hence not uniquely: different choices of $H$ correspond to different extensions which are, however, all isomorphic.
However, local non linear fields are not left invariant under the aforementioned isomorphisms: this results in a freedom in the definition of local fields like $T_{\mu\nu}$ which is nothing but the known regularization freedom \cite{HW01,HW05}.

As a final remark we would like to stress that the extended algebra can be tested on a smaller class of states: certainly, Hadamard states $\omega$ are contained within this class \cite{Radzikowski, BFK} and one can choose the two-point function $\omega_2$ of such a state as Hadamard function.
The study of the form of the stress tensor of a scalar field can be found in \cite{Mo03} or in \cite{HW05}.

\subsection{Restriction on timelike curves}

We are interested in restricting certain local observables like the point like product of fields on regular timelike curves $\gamma\subset M$.
The possibility of considering the restriction of local fields on timelike curves was discovered 
in \cite{Borchers}. Some further work in this direction can be found in \cite{keyl}. Recently similar ideas have been employed in \cite{Igor1, Igor2}. 
We would like also to remind the reader that the restriction of the expectation value of the stress tensor on timelike curves is used to derive bounds for the local energy density in curved space times, as is shown in \cite{Ford} or in a recent analysis in the realm  \cite{FewsterPfenning}. 

Here we would like to show that the restriction of fields on timelike curves, can be realized within $\mathcal{F}(M)$, in the sense that a certain set of functionals over smooth field configurations on $\gamma$ forms a sub algebra of $\mathcal{F}(M)$ defined above.
In order to clarify this point, we proceed introducing the following definition

\begin{definition}
Consider a timelike regular curve $\ga$ parametrized in such a way that $\dot{\ga}$ is future directed. The {\bf set of regular functionals} over smooth field configurations $\varphi\in C^\infty(\ga)$ is denoted by
\begin{equation}\label{eq:spetral-condition-gamma}
\mathcal{F}(\gamma) = \left\{ F\;, \text{differentiable with compact support}\;,\quad 
WF(F^{(n)}(\varphi)) \cap (\overline{\bR_+}^n\cup \overline{\bR_-}^n) = \emptyset   \right\}.
\end{equation}
where $\bR_+ = \{k\in\bR\;,\; k \geq 0\}$ is the set of positive directions and $\bR_- $ the set of negative ones.
\end{definition}

\noindent
Consider now the map $\iota_\gamma$ which realizes the restriction of the field configurations
\begin{equation}\label{eq:restriction}
\iota_\gamma: C^\infty(M) \to C^\infty(\gamma)\;,\qquad \iota_{\gamma} \varphi := \varphi\circ \gamma\;.
\end{equation}
We want to show that the pullback $\iota_\gamma^*$ on functionals over such smooth field configurations maps elements of $\mathcal{F}(\ga)$ to $\mathcal{F}(M)$.
We have actually the following proposition.

\begin{proposition}\label{pr:Fgamma-subset-FM}
Let $\ga$ be a smooth timelike curve, consider $F\in \mathcal{F}(\gamma)$
and $\iota_\gamma$ defined in \eqref{eq:restriction}, then  $\iota_\gamma^* F$ is contained in $\mathcal{F}(M)$.
\end{proposition}
\begin{proof}
Since the curve $\ga$ is smooth, we have that, for every $\varphi\in C^\infty(M)$, $\iota_\gamma\varphi$ is also smooth.
Furthermore, since $\iota_\gamma$ is also linear, $\iota_\gamma^*F$  is differentiable and thus we obtain an explicit formula 
$$
\iota_\gamma^* F^{(n)}(\varphi)(h_1\otimes\dots\otimes h_n) =  F^{(n)}(\iota_\gamma \varphi) (\iota_\gamma h_1\otimes\dots\otimes \iota_\gamma h_n)\;.
$$
Let us proceed discussing its support.
Notice that the support of $\iota_\gamma^*F^{(n)}$ is contained on $\gamma^{\otimes n}\subset M^n$ and there coincides with the image of the support of $F^{(n)}$ under the continuous map $\gamma^{\otimes n}$, hence it is compact.
In order to conclude the proof, we need to show that the wave front set of $\iota_\gamma^*F^{(n)}$ is disjoint from $\overline{V_+}^n\cup \overline{V_-}^n$.
We address this problem looking for bounds for the wave front set of 
$F^{(n)}(\iota_\gamma \varphi) (\iota_\gamma h_1\otimes\dots\otimes \iota_\gamma h_n)$
interpreted as the composition of $F^{(n)}(\iota_\gamma \varphi): \mathcal{D}(\gamma^{\otimes n} )\to \mathbb{C}$ with $\iota_\gamma^{\otimes n}$.

Notice that, for every $h\in \mathcal{D}(M)$, $\iota_\gamma(h)$ is a linear transformation from $\mathcal{D}(M)$ to $\mathcal{D}(\gamma)$. We indicate by $K$ its integral kernel which is thus an element of $\mathcal{D}'(M\times\gamma)$ and 
\begin{equation}\label{eq:WF(K)}
WF(K) =\left\{(x,t,k,\xi) \in T^*(M\times \gamma)\setminus \{0\}\;,\; \gamma(t)=x  \;,\; (d\gamma_{t} \xi - k)(\dot\ga) = 0    \right\}
\end{equation}
where $d\gamma$ is the differential map, hence, $d\gamma_t \xi$ is the push forward of $\xi$ to $T_{\gamma(t)}M$.
The distribution $\iota_\gamma^*F^{(n)}(\varphi)$ is then constructed as the composition of
$F^{(n)}(\iota_\gamma\varphi)$ with the multi linear map $\iota_\gamma^{\otimes n}$, along the line of theorem 8.2.13 of \cite{Hormander}.
Using the notation introduced therein, since 
$WF(K^{\otimes n})_{\gamma^n}$ is empty, the hypothesis of that theorem are fulfilled.
Thus the composition is well defined and its wave front set is such that
$$
WF(\iota_\gamma^*F^{(n)}) \subset WF(K^{\otimes n})_{M^n} \cup WF'(K^{\otimes n}) \circ WF(F^{(n)})\;. 
$$
Let us analyze carefully its form. First of all, notice that $WF(K^{\otimes n})_{M^n}$ contains only spatial directions, while, from \eqref{eq:spetral-condition-gamma}, $WF(F^{(n)})\cap (\overline{\bR_+}^n\cup \overline{\bR_-}+^n) = \emptyset$. 
In order to conclude the proof, we observe that the composition of $WF(F^{(n)})$ with $WF'(K^{\otimes n}) = \{ WF'(K) \cup  \{0\}   \}^{\otimes n} \setminus \{0\}$ 
where  $WF(K)$ is given in \eqref{eq:WF(K)},
 has disjoint intersection with $\overline{V_+}^n\cup \overline{V_-}+^n$.
\end{proof}

The previous proposition guaranties that the set $\mathcal{F}(\gamma)$ is contained in $\mathcal{F}(M)$, we shall now proceed introducing an algebraic structure on this set. 
Actually, we would like to equip $\mathcal{F}(\gamma)$ with a non commutative product $\star_R$ whenever $R\in \mathcal{D}'(\gamma\otimes\gamma)$, its symmetric part is real while its antisymmetric one is imaginary  and
$$
WF(R) \subset \left\{(x_1,x_2,k_1,k_2) \in T^*\bR^2\setminus\{0\}\;,\; x_1=x_1\;,\;   k_1=-k_2,\ k_1>0 \right\}\;.
$$
We shall consider a product defined as
\begin{equation}\label{eq:prod-Fgamma}
F\star_{R} G(\varphi) = F(\varphi)G(\varphi)  + \sum_{n\geq 1} \left\langle F^{(n)}(\varphi)   , R^{\otimes n} G^{(n)}(\varphi)  \right\rangle_n
\end{equation}
where $\langle \cdot,\cdot\rangle_n$ is the standard pairing in $T^*\bR^n$.
With this product $(\mathcal{F}(\gamma),\star_R)$ acquires the structure of a $*$-algebra if equipped with the $*$-operation corresponding to complex conjugation
$$
F^*(\varphi) = \overline{F(\overline{\varphi})}\;.
$$

We would like to show that whenever $R$ is chosen as the map of $H$ under the pullback 
$(\gamma\otimes\gamma)^*$ the product in $\mathcal{F}(\gamma)$ agrees with that of $\mathcal{F}(M)$.
We start with the following proposition:
\begin{proposition}\label{th:H-restriction}
Let $\gamma$ be a timelike regular curve. The pull-back $(\gamma\otimes\gamma)^*H$ of any Hadamard function $H$ on 
$\mathcal{D}(\gamma\otimes\gamma)$ is unambiguously defined and its wave front set is 
$$
WF((\gamma\otimes\gamma)^*H) = \left\{(t_1,t_2,k_1,k_2) \in T^*\bR^2 \setminus \{ 0\} \;,\;t_1=t_2\;,\; k_1=-k_2\;,\; k_1>0  \right\}
$$
\end{proposition}
\begin{proof}
In order to prove that the restriction of any $H$ on $\ga(\bR)\otimes \ga(\bR)$ is well defined, 
we shall show that the hypothesis of theorem 8.2.4 in \cite{Hormander}  are fulfilled.
Consider a regular parametrization of the curve $\ga:\bR\to M$. 
We shall hence build a map $f$ from $\bR^2\to M^2$ as $f(t_1,t_2) = (\ga(t_1),\ga(t_2))$.
The set of normals to that map is
$$
N_f = \left\{(\ga(t_1),\ga(t_2),k_1,k_2)\in T^*M^2\;,\; k_1(\dot\ga(t_1))=k_2(\dot\ga(t_2))=0\right\}\;.
$$
Notice that, since $\dot\ga(t)$ is a non vanishing timelike vector for all $t$, if $(t_1,t_2,k_1,k_2)$ is in $N_f$ then either $g^{-1}k_1$ or $g^{-1}k_2$ is spacelike.
All the covectors in $WF(H)$ are null, hence, 
we conclude that, 
$$
WF(H) \cap N_f = \emptyset\;.
$$
Hence by theorem 8.2.4 in \cite{Hormander}, the pullback of $H$ under $f^*$ is unambiguously defined and furthermore, 
$$
WF(f^*H) \subset f^* WF(H) = \left\{(t_1,t_2,k_1,k_2) \in T^*\bR^2 \setminus \{ 0\} \;,\;t_1=t_2\;,\; k_1=-k_2\;,\; k_1>0  \right\}\;.
$$
In order to conclude the proof, we need to show the other inclusion.
We notice that, for every choice of $t_1=t_2$, there is only one singular direction in $f^* WF(H)$. Thus, in order to have $WF(f^*H) \supset f^* WF(H)$, we just need to show that the singular support of $f^*H$ is formed by the set $D:=\{(t_1,t_2)\;, t_1=t_2\}$. The singular support of $f^*H$ surely contains the singular support of the causal propagator $\Delta$ which is proportional to the antisymmetric part of $H$. 
The singular support of $f^*\Delta$ is $D$.
\end{proof}

Thanks to the previous proposition, we have that $R:=(\gamma\otimes\gamma)^*H$ can be used in \eqref{eq:prod-Fgamma} to construct a product in $\mathcal{F}(\gamma)$.
Moreover, this choice implies that the two products $\star_H$ in $\mathcal{F}(M)$ and $\star_{R}$ in $\mathcal{F}(\gamma)$ agrees. 
Actually, we have the following proposition,

\begin{proposition}
Let $\gamma$ be a regular timelike curve, $\mathcal{F}(\ga)$ is a sub $*$-algebra  of  $\mathcal{F}(M)$.
\end{proposition}
\begin{proof}
Proposition \ref{pr:Fgamma-subset-FM} implies that as a set $\mathcal{F}(\ga) \subset \mathcal{F}(M)$. 
Analyzing the form of \eqref{eq:prod-Fgamma} and thanks to the choice of $R$ as the restriction of $H$ on $\gamma\otimes\gamma$ we have that 
$$
\iota^*_\gamma(F\star_R G) = \iota^*_{\gamma}F\star_H \iota^*_{\gamma}G\;.
$$ 
We can conclude the proof, noticing that 
the restriction preserves complex conjugation, and hence   
$$
(\iota^*_\gamma F)^* = \iota^*_\gamma (F^*)
$$ 
for every $F$ in $\mathcal{F}(\ga)$.
\end{proof}
In order to simplify the presentation,  from now on, we shall not explicitly indicate $\iota$. 
For the same reason we shall denote the product in $\mathcal{F}(\gamma)$ by $\star_\omega$ namely, with the same symbol of the product in $\mathcal{F}(M)$.

\section{Deviation equation for congruences of timelike-geodesics}\label{sec:3}

Let us fix a background spacetime $(M,g)$ where $M$ is a smooth manifold and 
$g$ a background metric. Consider a congruence of timelike geodesics  on this background. 
We want to study what is the probability of having a focusing point induced by the quantum nature of $\varphi$. To this end we shall first of all remind how this problem is treated classically.

Let $\xi$ be the vector field of tangents to the congruence parametrized by proper time.
As usual, we start discussing the infinitesimal spatial displacement of the tensor field defined as  $B_{ab} = \nabla_b \xi_a$. 
This tensor field can be decomposed into simpler fields, namely the
{\bf expansion} $\theta = \nabla_b \xi^b$, the {\bf shear} $\sigma_{ab}$ and the {\bf twist} $\omega_{ab}$. The geodesic deviation equation induces some equations for these fields.
For a complete discussion we refer to \cite{Wald,HawkingEllis}, here we are interested only  in the Raychaudhuri equation namely the equation for the expansion parameter
\begin{equation}\label{eq:raychaudhuri}
\xi^c \nabla_c \theta = -\frac{1}{3} \theta^2 - \sigma_{ab} \sigma^{ab} +\omega_{ab}\omega^{ab} - R_{cd} \xi^c \xi^d\;,
\end{equation}
which can also be written as 
\begin{equation}
\dot\theta = -\frac{1}{3} \theta^2 - \sigma_{ab} \sigma^{ab} +\omega_{ab}\omega^{ab} - \left(T_{ab} 
-\frac{1}{2}g_{ab} T \right) \xi^a \xi^b\;.
\end{equation}
We shall use it to obtain the passive influence of matter fluctuations on the expansion parameter. In this process, as discussed in the introduction, we shall assume that both the shear and the twist of the geodesic congruences are not affected by those fluctuations, namely we shall assume that they are classical fields. 
This approximation can be justified analyzing the full deviation equation and noticing that, at least when both $\sigma_{ab}$ and $\omega_{ab}$ vanish on the background, their linear perturbations over the chosen solutions manifest as second order effects in \eqref{eq:raychaudhuri}.

When the matter content of the theory is described by a massless minimally coupled field indicated by $\varphi$, equation \eqref{eq:raychaudhuri} can be further  simplified.
Using $\theta = 3\frac{\dot\phi}{\phi}$, we get a linear equation in $\phi$
$$
 \ddot\phi + \frac{1}{3}\left(  \sigma_{ab} \sigma^{ab} -\omega_{ab}\omega^{ab} + \frac{1}{2}T_{anomalous}\right) \phi 
 +\frac{1}{3}  \dot \varphi  \dot \varphi \;\phi = 0\;,
$$
where we have indicated by $T_{anomalous}$ the quantum anomalous contribution to the stress tensor.
This term is responsible for the trace anomaly of the stress tensor \cite{Wald78}.
Further details on this term can be found in \cite{Mo03} or in \cite{HW05}. 

In the next section, we would like to check when there is a focusing of the geodesics forming the congruence. Notice that when $\theta$ diverges such a focusing occurs. A sufficient condition for this to occur is the vanishing of $\phi$.

\subsection{Perturbative analysis of the Raychaudhuri equation}

We shall here discuss the perturbative solution of the equation 
\begin{equation}\label{eq:quantum-one-interaction-YF}
 \ddot\phi + V \phi 
 +\frac{1}{3} \lambda \dot \varphi  \dot \varphi \;\phi = 0
\end{equation}
where, once a congruence is chosen, the classical external potential is 
$$
V:=\frac{1}{3}\left(  \sigma_{ab} \sigma^{ab} -\omega_{ab}\omega^{ab} + \frac{1}{2}T_{anomalous}\right) \;,
$$
which depends only on the background metric. 
We shall assume that $V$ is a smooth function which vanishes at infinity, as it happens in a large class of physical situations.

Moreover,  $\lambda=1$ plays the role of a coupling ``constant''.
The interesting part in the previous equation is in $\dot\varphi\dot\varphi \phi$ which represents a quantum interaction between the field $\varphi$ and $\phi$.
Furthermore, in order to isolate infrared problems, we shall consider initially 
$\lambda$ in $\mathcal{D}(\mathbb{R})$ and later we shall discuss the adiabatic limit, namely the limit $\lambda\to 1$.

In the approximation we are employing, as discussed above, we discard the influence of $\phi$ on $\varphi$, which means that neither the algebra $\mathcal{F}$ nor the states $\omega$ on that algebra are influenced by $\theta$.
With this in mind, up to the choice of some initial conditions or better a classical solution for the free case, the previous equation can be solved within $\mathcal{F}(\gamma)$. Later on, once a suitable state $\omega$  is chosen, we can obtain the $n$-point functions of $\phi$ out of the expectation values of $\varphi$.

The aim of the present section is to obtain  an expression for the solution $\phi$ in a perturbative series, proving that each term of the sum defines an element of the algebra $\mathcal{F}(\gamma)$.
When both $\dot\varphi$ and $\phi$ are classical smooth functions, equation \eqref{eq:quantum-one-interaction-YF} can be written in integral form, namely as a sort of Yang-Feldman equation 
\begin{equation}\label{eq:yang-feldmann}
\phi-\phi_0 +R_V(\lambda  \dot \varphi  \dot \varphi \;\phi) = 0
\end{equation}
where $\phi_0$ is a solution of the free equation and $R_V$ is the unique retarded fundamental solution associated to
\begin{equation}\label{eq:eqdiff1d}
P_\gamma=\frac{d^2}{dt^2} +V\;.
\end{equation}
The retarded fundamental solution is an operator from $\mathcal{D}(\bR)\to C^\infty$ such that, for every $f\in\mathcal{D}(\bR)$,
$$
P_\gamma R_V(f) = f\;
$$
and $\text{supp}(R_V(f))$ is contained in the future of the support of $f$.
This operator exists and it is unique, provided $V$ is sufficiently regular as is the case for the potentials we are considering in this section.
Moreover, the integral kernel associated to that map is of the form  
\begin{equation}\label{eq:RV}
R_V(x_1,x_2) = \vartheta(x_1-x_2) S(x_1,x_2)\qquad
R_Vf(x) = \int R_V(x,y)f(y)dy,
\end{equation}
where $\vartheta$ is the Heaviside step function and $S$ is the unique smooth real bi-solution of \eqref{eq:eqdiff1d} such that $S(x,x) = 0$, $S(x_1,x_2) = - S(x_2,x_1)$
and $\frac{\partial{S}}{{\partial x_1}}(x,x) = 1$. 

The solution $\phi$ of \eqref{eq:yang-feldmann} can be treated as a formal power series in the coupling parameter $\lambda$.
$$
\phi = \sum_{n=0}^\infty \phi_n
$$
where $\phi_n$ are homogeneous expression of degree $n$ in $\lambda$
and $\phi_0$ is a classical solution.
Hence,   
$$
\phi_n = -R_V(\dot\varphi^2 \phi_{n-1})\;.
$$

\subsection{Perturbative algebraic quantum field theory and renormalization}

Aim of this section is to show that, when $\varphi$ is treated as a quantum field, the perturbative analysis of \eqref{eq:yang-feldmann} can be performed at the algebraic level. 
In the literature, there are many approaches to the perturbative analysis of quantum field theories, here,
we shall follow a procedure similar to that developed in the framework of perturbative algebraic quantum field theory \cite{BF00, BDF} (see also the lecture notes \cite{BF09,FredenhagenRejzner}).
However, since the influence of curvature fluctuations on the matter field will be neglected,
the analysis will be much simpler, and thus we shall readapt the previously known results to that case. 
Furthermore, instead of constructing the interacting field by means of the Bogoliubov formula, we shall employ directly the perturbative solution of \eqref{eq:yang-feldmann}.
We shall treat certain types of ``retarded'' products of matter fields instead of time ordered ones. 
Despite of this difference, we shall see that it is possible to employ an inductive construction in the number of fields, similar to the one employed for time ordered products \cite{EG, BF00}. 
However, at every inductive step, these products are known up to the thin diagonal. 
We shall employ known renormalization techniques to extend those distributions 
to the thin diagonal preserving their scaling degree \cite{BF00}. 
We shall furthermore stress that the products employed in the present paper are slightly different than the 
 retarder product present in the literature, see \cite{glz, DF} for a definition.

We shall again treat the solution of equation \eqref{eq:yang-feldmann} as a formal power series in the coupling parameter $\lambda$.
We would like to show that the interacting fields, treated as formal power series in $\lambda$, are contained in $\mathcal{F}(\gamma)$.
Let us start noticing that, at least formally, choosing $\lambda \in \mathcal{D}(\mathbb{R})$ and using the $\star_\omega$ product induced by an Hadamard state $\omega$ we have that
\begin{gather}
\boldsymbol{\phi}_n(f) ``:=" (-1)^n
\int f(x_n) R_V(x_n,x_{n-1}) R_V(x_{n-1},x_{n-2})\dots  R_V(x_1,x_{0})  \;  \left( \lambda(x_{n-1})\dots \lambda(x_0)\right) \cdot  \;
\nonumber
\\
\cdot\left(\dot\varphi^2(x_{n-1}) \star_\omega \dots \star_\omega \dot\varphi^2(x_0) \right)         \phi_0(x_0)  \; dx_0\dots dx_n
\label{eq:phin}
\end{gather}
where we used a compressed  notation $\dot\varphi^2 \star_\omega \dots \star_\omega  \dot\varphi^2$
for the integral kernels of the distributions on $\mathcal{D}(\bR)^{\otimes n}$ 
$$
\dot\varphi^2 \star_\omega \dots\star_\omega \dot\varphi^2  (f_1\otimes\dots\otimes f_n)  :=   \dot{\boldsymbol{\varphi}}^2(f_1) \star_\omega \dots \star_\omega \dot{\boldsymbol{\varphi}}^2(f_n)\;.
$$
Actually, if we indicate by $T_f$ the element of $\mathcal{F}(\gamma)$ defined as
$$
T_f(\varphi) := \int_\bR f \dot\varphi^2 dx\;,
$$ 
a more precise definition of those integral kernels is
\begin{gather*}
\dot\varphi^2(x_{n}) \star_\omega \dots \star_\omega \dot\varphi^2(x_1):=
\frac{\delta^{n}}{\delta f_{n}(x_{n})\dots\delta{f_1}(x_{1})}T_{f_{n}}\star_\omega \dots \star_\omega T_{f} \;.
\end{gather*}
In order to make the expression presented in \eqref{eq:phin} precise, we need to discuss the products of the retarded distributions $R_V$ with the expectation values of $\omega$. However, they are in general too singular to be multiplied. The resulting product is thus well defined only outside of coinciding limit. 
Hence, we notice that  when the points $x_0$,\ldots,$x_n$ are all disjoint, it is possible to safely multiply $R_V$s with $\dot{\varphi}$s: furthermore, in that case we have formally
\begin{gather*}
R_V(x_{n-1},x_{n-2})\dots  R_V(x_1,x_0)\ \dot\varphi^2(x_{n-1}) \star_\omega \dots \star_\omega \dot\varphi^2(x_0)
:= \\
S(x_{n-1}-x_{n-2})\dots  S(x_1,x_{0})  \; r(x_{n-1},\dots , x_0)\qquad n\geq 1,
\end{gather*}
where we used the explicit form of $R_V$ given in \eqref{eq:RV} and 
$\{r(x_n,\dots x_0)\}_{n\geq 0}$ are certain distributions we are going to discuss.
Since $S$ are smooth functions, the problem that we have to deal with is about the definition of $r({x_n},\dots,x_0)$ when some points coincide.
In the next, we construct the $r$ with an inductive procedure, similar to the one used to construct time ordered products for interacting quantum fields. 
In order to employ this procedure, we enumerate here the properties of these $r$ we shall use.
The functional valued distributions $r(x_n,\ldots,x_0)$ are well defined outside the diagonals and satisfy the following properties:
\begin{itemize}
\item{\textbf{Functional equation}:} 
\begin{equation}\label{eq:def-r}
r({x_n},\dots,x_0)=
\frac{\delta^{n}}{\delta f_{n}(x_{n})\dots\delta{f_0}(x_{0})}T_{f_{n}}\star_\omega \dots \star_\omega T_{f_0}
\end{equation}
when all points $x_n,\ldots,x_0$ are strictly decreasing.

\item{{\bf Retardation:}} $r(x_n,\dots,x_0)=0$ if the points are not ordered, namely if does not hold $x_n\geq \dots \geq x_0$.
\item{{\bf Factorization:}} if $x_n\geq \dots \geq x_0$ and $x_m>x_{m-1}$ for some $0<m\leq n$ then
$$
r(x_n,\dots,x_0) =  r(x_n,\dots,x_m)\star_\omega r(x_{m-1},\dots,x_0) 
$$
\item{{\bf Initial element:}} $r(x_0)=\dot{\varphi}^2(x_0)$.
\end{itemize}
We will use this characterizing properties in order to extend the retarded products $r$ on the full space.
Once we have constructed all the $r$, we will have that  
\begin{equation}\label{eq:phi0e1-r}
\boldsymbol{\phi}_0(f) := \int f(x_0) {\phi}_0(x_0)\;   dx_0\;,\qquad 
\boldsymbol{\phi}_1(f) := -\int f_R(x_0) \phi_0(x_0)\lambda(x_0)\dot{\varphi}^2(x_0)\;  dx_0 
\end{equation}
where $f_R$ is the smooth function obtained applying the adjoint of the retarded operator to the smooth function $f$, formally 
$$
f_R(x):=\int f(y)R_V(y,x)  dy\;.
$$
Finally, for $n$ bigger than $1$, we will have
\begin{gather}
\boldsymbol{\phi}_n(f) = (-1)^n
\int f_R(x_{n-1}) S(x_{n-1},x_{n-2}) S(x_{n-2},x_{n-3})\dots  S(x_1,x_{0})  \;  \left( \lambda(x_{n-1})\dots \lambda(x_0)\right) \;  \phi_0(x_0)\;
\nonumber
\\
r(x_{n-1},\dots, x_{0})  \; dx_0\dots dx_n\;.
\label{eq:phin-r}
\end{gather}

\subsection{Inductive procedure up to the total diagonal}\label{sec:inductive}

Employing the functional equation and the retardation property of $r$ stated above, we can construct them up to the 
diagonals. 
Here, we shall see that
the elements $r$ can be constructed by means of an inductive procedure over the number of entries in a way similar to the analysis performed by Epstein and Glaser \cite{EG}.
In particular, we shall assume that we have fully constructed the $r$s with $n$ entries, we shall hence show
that, thanks to the factorization property, it is possible to construct $r$ with $n+1$ entries up to the thin diagonal $\Delta_{n+1}$
namely the set of elements $(x_n,\dots, x_0)\in \mathbb{R}^{n+1}$ such that $x_0=x_1=\dots=x_n$.
The step which is then missing is the extension on the full space.
We shall discuss the last step in the next section, here, we shall use the factorization property to construct $r$  in
$\mathcal{D}'(\mathbb{R}^{n+1}\setminus \Delta_{n+1};\mathcal{F})$ having already constructed 
$r\in\mathcal{D}'(\mathbb{R}^{m};\mathcal{F})$ for every $m\leq n$.

Let us start discussing the patching of $\mathbb{R}^{n+1}$ up to the total diagonal in chart where $r$ can be factorized.
In accomplishing this step we shall adapt the analysis presented in  \cite{BF00} for the time ordered products to the retarded one and hence the number of charts we shall use will be much smaller. 
The set $\mathcal{S}^c$ is the complement of the set $\mathcal{S} = \{(x_n,\dots,x_0)\;,x_n\geq \dots \geq x_0\}$: due to causality $r$ vanishes on $\mathcal{S}^c$.
The other charts we shall use are $\mathcal{C}_m$ for $m\in\{1,\dots,n\}$, where $\mathcal{C}_m$ is defined as the open set of points $(x_n,\dots, x_0)\in \mathbb{R}^{n+1}$ where $x_i > x_j$ for every $i < m$, and $j \geq m$.

\begin{proposition}
Up to the small diagonal the space $\mathbb{R}^{n+1}$ is covered by the union of all the $\mathcal{C}_m$ and by $\mathcal{S}^c$ i.e.
$$
\mathbb{R}^{n+1} \setminus \Delta_{n+1}  = \mathcal{S}^c\cup\bigcup_{m\in[1,n]} \mathcal{C}_m 
$$
\end{proposition}
\begin{proof}
Clearly $\mathcal{S}^c$ and every $\mathcal{C}_m$ are contained in $\mathbb{R}^{n+1}\setminus \Delta_{n+1}$. In order to show the other inclusion, let's take a generic element $X$ of $\mathbb{R}^{n+1} \setminus \Delta_{n+1}$, either $X$ is $\mathcal{S}$ or in its complement. 
We want to show that if $X$ is not in $\mathcal{S}^c$ than it must be contained in $\mathcal{C}_m$ for some $m$. So suppose that $X=(x_n,\dots,x_0)$ is in $\mathcal{S}$, than $x_n\geq \dots \geq x_0$. However, since $X\not\in\Delta_{n+1}$, one of those inequalities must be strict, namely it must exist an $m$ bigger than $1$ and $m$ smaller or equal to $n$ for which $x_{m-1}> x_m$. Hence $X\in\mathcal{C}_m$.
\end{proof}
We shall now discuss the form of $r$ on the patches presented in the previous proposition. 
First of all, we notice that, due to the retardation property, on $\mathcal{S}^c$ the retarded product $r$ must vanish. At the same time, on any $\mathcal{C}_m$ we might define $r$ as following.

\begin{definition}
On $\mathcal{C}_m$, we define  $r_m(x_n,\dots,x_0):= \left.r\right|_{\mathcal{C}_m}(x_n,\dots,x_0)$ as 
the $\star_\omega$ product of two $r$ with a smaller number of entries, namely
\begin{equation*}
r_m(x_n,\dots,x_0) :=  r(x_n,\dots,x_m) \star_\omega r(x_{m-1},\dots,x_0),\quad r(x_0):=\dot{\varphi}^2(x_0);.
\end{equation*}

\end{definition}

In order to be allowed to use these elements consistently we need to 
check that on the overlap, of different $\mathcal{C}_m$ the different factorization coincide.
\begin{proposition}
Different factorizations coincide on the overlap, namely
$$
\left.r_{m_1}\right|_{\mathcal{C}_{m_1}\cap\mathcal{C}_{m_2}}
=
\left.r_{m_2}\right|_{\mathcal{C}_{m_1}\cap\mathcal{C}_{m_2}}\;.
$$
and
$$
\left.r_{m}\right|_{\mathcal{C}_{m}\cap\mathcal{S}^c}=0\;
$$
for any $m_1,m_2,m\in\{1,\ldots,n\}$. 
\end{proposition}
\begin{proof}
Consider, first of all, the points in $(x_n,\dots,x_0) \in \mathcal{C}_{m_1}\cap\mathcal{C}_{m_2}$, and suppose, without loosing generality, that $m_1<m_2$. Hence, due to retardation property,
$$
r_{m_2}(x_n,\dots, x_0) = r(x_n,\dots,x_{m_2}) \star_\omega r(x_{m_2-1},\dots, x_0) \;.
$$
Since $m_1<m_2$ we can further factorize the last factor in the right hand side of the previous expression, namely
$$
r_{m_2}(x_n,\dots, x_0) = r(x_n,\dots,x_{m_2}) \star_\omega r(x_{m_2-1},\dots, x_{m_1}) \star_\omega r(x_{m_1-1},\dots, x_0) \;.
$$
Since the star product is associative, we can conclude the first part of the proof, namely we have
$$
r_{m_2}(x_n,\dots, x_0) = r(x_n,\dots, x_{m_1}) \star_\omega r(x_{m_1-1},\dots, x_0) 
=r_{m_1}(x_n,\dots, x_0) \;.
$$
Now let us consider the points in $\mathcal{C}_m\cap \mathcal{S}^c$, hence, in there,
$$
r_m(x_n,\dots,x_0) = r(x_n,\dots,x_{m}) \star_\omega r(x_{m-1},\dots, x_0)\;.
$$
However, $(x_n,\dots,x_0)\in\mathcal{S}^c$, so at least one of the  inequality 
$x_n\geq\dots \geq x_0$ does not hold.
Since it is still true that $x_m>x_{m-1}$, either $r(x_n,\dots,x_{m})$ or $r(x_{m-1},\dots, x_0)$ vanish.
\end{proof}
We are now ready to glue together all the $r_m$ introduced above to a distribution with all the desired property defined outside $\Delta_{n+1}$.
\begin{definition}
On $\mathbb{R}^{n+1}\setminus\Delta_{n+1}$, we define $^0r(x_n,\dots,x_0)$ as 
$$
^0r(x_n,\dots,x_0) = \sum_m f_m(x_n,\dots,x_0) r_m(x_n,\dots,x_0)\;,
$$
where $\{f_m\}_{0\leq m\leq n}$ is a smooth partition of unity on $\mathbb{R}^{n+1}\setminus \Delta_{n+1}$ adapted to 
$\{\mathcal{S}^c,\mathcal{C}_1,\dots,\mathcal{C}_n\}$. 
\end{definition}
Hence, on every $X$ in $\mathbb{R}^{n+1}\setminus \Delta_{n+1}$ it holds that $\sum_{m=1}^n f_n(X) = 1$, furthermore $f_m$ is supported on $\mathcal{C}_m$ for $m\geq 1$ and on $\mathcal{S}^c$ otherwise: the existence of such a partition of unity is an immediate consequence of the finiteness of the covering of $\mathbb{R}^{n+1}\setminus\Delta_{n+1}$.

By definition the functional valued distributions $^0r(x_n,\ldots,x_0)$ satisfy the properties of functional equation, retardation, factorization and initial element.
Furthermore the inductive construction described above shows that even the procedure of extension of $^0r$ on the diagonals could be done inductively: in fact, assuming to have fixed an extension of $^0r(x_p,\ldots,x_0)\forall p=1,\ldots,n-1$, by the factorization property the extension of $^0r(x_n,\ldots,x_0)$ should be done only on the full diagonal $\Delta_{n+1}$.


\subsection{Extension to the thin diagonal $\Delta_{n+1}$}\label{sec:extension}

In this subsection, we shall discuss the extension of the distributions $^0r(x_n,\dots,x_0)$ obtained in the preceding section to the thin diagonal $\Delta_{n+1}$.
This will complete the inductive construction of $r$.
As discussed above, given a set of $r(x_p,\ldots,x_0)$ $\forall p=0,\ldots,n-1$, we can construct $^0r(x_n,\ldots,x_0)$ which is defined up to the thin diagonal, namely it is an element of $\mathcal{D}'(\mathbb{R}^n\setminus \Delta_{n+1})$; in order to complete the inductive step, we need to extend it on $\Delta_{n+1}$.
For this we proceed exactly as in \cite{Steinmann,BF00}: if the scaling degree of $^0r$ towards $\Delta$ is finite, we may apply either Theorem 5.2 or Theorem 5.3 of \cite{BF00} to get the desired extension.
The scaling degree towards all the diagonals is always finite, so we can extend all these distributions preserving the scaling degree: in this procedure there is certain renormalization freedom, which can be fixed at all orders.
As a consequence, the procedure to extend $r(x_n,\dots,x_0)$ to the thin diagonal is not unique, however, the ambiguities in this extension are all local and they are thus characterized by a finite set of numbers. This freedom is usually called {\bf renormalization freedom} and it can be restricted applying some physical criteria like translation invariance. The residual freedom is usually fixed in accordance with experiments.
%

After employing the renormalization conditions as the one discussed above, we notice that the following proposition holds:
\begin{proposition}
The functionals $\boldsymbol{\phi}_n(f)$ constructed in \eqref{eq:phin-r} and in \eqref{eq:phi0e1-r} are elements of $\mathcal{F}(\ga)$.
\end{proposition}
\begin{proof}
In order to prove this assertion we have to analyze the wave front set of the functional derivatives of $\boldsymbol{\phi}_n(f)$, namely, it must be such that
\begin{equation}\label{eq: condizione sul WF}
WF(\boldsymbol{\phi}^{(m)}_n(f))\cap(\overline{\mathbb{R}^m_+}\cup\overline{\mathbb{R}^m_-})=\emptyset.
\end{equation} 
First of all we notice that only a finite number of functional derivative of $\boldsymbol{\phi}_n$ do not vanish 
$$
\boldsymbol{\phi}_n^{(m)}(f) =0 \qquad \text{for}\qquad m>2n\;.
$$ 
The $m$-th derivative of $\boldsymbol{\phi}_n$ contains a linear combination of $r$ with a different number of arguments which are sometime contracted with smooth functions ($S$ or $\lambda$ in \eqref{eq:phin-r}): by linearity of the wave front set, and by a straightforward application of theorem 8.2.12 in \cite{Hormander}, it is sufficient to prove that each $r$ satisfies \eqref{eq: condizione sul WF}.
Let us start analyzing the one with two entries: 
$
r(x_1,x_0)
$ which is the first one with non empty wave front set.
We see that, outside the thin diagonal $\Delta_2$, $r={^0r}$ which is either zero or $\dot{\varphi}^2(x_1)\star_\omega\dot{\varphi}^2(x_0)$ and in both cases the property \eqref{eq: condizione sul WF} is satisfied.
Regarding its extension $r$, we have that  $WF(r)$ is contained within the closure in space of $WF(^0r)$ (to cover also $\Delta_2$) joined with the wave front set of the 
additional terms imposed by the employed renormalization techniques which are proportional to $\delta(x_1,x_0)$ or its derivative.
The wave front set of $\delta(x_1,x_0)$ has empty intersection with $\overline{\mathbb{R}}^2_+\cup\overline{\mathbb{R}_-^2}$ because
$WF(\delta(x_1,\ldots,x_n))=\{(x_1,\ldots,x_n,k_1,\ldots,k_n)|\ x_1=\ldots=x_n,\ k_1+\ldots+k_n=0\}$, property \eqref{eq: condizione sul WF} is then satisfied by $r(x_1,x_0)$.
Since all other $r$ are constructed by induction, in order to conclude the proof, we just need to check the $n-$th step.
We see that $r$ with $n+1$ entries, outside the thin diagonal is constructed by means of the factorization property and the $\star_\omega$ product among two $r$ enjoys the property \eqref{eq: condizione sul WF}.
At the last step we notice that \eqref{eq: condizione sul WF} holds also on $\Delta_{n+1}$, because 
of the $\delta$-form of the terms which contain the renormalization constants.
\end{proof}

\subsection{Adiabatic limit}\label{sec:adiabatic-limit}

So far we have analyzed the ultraviolet problems present in the perturbative construction of $\phi$. We shall now discuss the adiabatic limit, namely the limit $\lambda\to 1$ and we shall see that no infrared problems arise. Furthermore, whenever $f$ is a compactly supported smooth function,
the limit $\lambda\to 1$ of every $\boldsymbol{\phi}_n(f)$ can be performed within $\mathcal{F}(\gamma)$. 
%
%
%
%
%
%
In order to prove this claim, we notice that, for every  classical solution $\phi_0$ of the free equation \eqref{eq:eqdiff1d} 
we want to perturb, 
 it is always possible to 
select a compactly supported smooth function $K$ whose support lies in the past of the support $f$ without intersecting $\supp f$ ($\supp K \subset J^-(\supp f)$ and $\supp f \cap \supp K = \emptyset$) and such that $R_V(K)$ equals $\phi_0$ in the future of the support of $K$, namely
\[
\phi_0(x) = \phi^+_0(x) := R_V(K)(x) \;,\qquad \forall x\in J^+(\supp K) \setminus \supp K \;.
\]
With this choice of $K$, since the equation \eqref{eq:yang-feldmann} is linear in $\phi$,
 we can formally split any solution 
 in two parts $\phi=\phi^++\phi^-$ where $\phi^\pm$ solve an inhomogeneous equation 
\[
\ddot{\phi}^\pm
+ V\phi^\pm
+\frac{1}{3}\dot{\varphi}^2\phi^\pm=\pm K,
\]
and where $\phi^\pm$ are supported in the future/past of the support of $K$ ($\supp \phi^\pm\subset J^\pm(\supp K)$). 
The solution $\phi^+(f)$ can be obtained perturbatively\footnote{The field $\phi^-(f)$ can be obtained  perturbatively employing advanced products of $\dot{\boldsymbol{\varphi}}^2$. The advanced products can be constructed in a similar way as the distributions $r$ discussed in the paper.}
to all order employing equations \eqref{eq:phi0e1-r} and \eqref{eq:phin-r} and replacing $\phi_0$ with $\phi^+_0 := R_V(K)$.     
Finally, thanks to the choice of the support of $K$, for every $n$, $\boldsymbol{\phi}_n(f)$ coincides with 
$\boldsymbol{\phi}_n^+(f)$, however now all the retarded integrals appearing in $\boldsymbol{\phi}_n^+(f)$ are performed on compact intervals  also when $\lambda$ is chosen to be equal to $1$.

\section{Perturbative expansion for the expansion parameter over a flat background}\label{sec:4}

In this section we shall apply the previously introduced general picture in order to estimate the behavior of a congruence of parallel timelike geodesics on a flat background,  
i.e. $(M,g)=(\mathbb{M}^4,\eta)$ Minkowski spacetime. 
Notice that, with a suitable choice of the renormalization conditions, the Minkowksi spacetime is a solution of the semiclassical Einstein equation whenever the stress tensor of the matter field $\varphi$ is considered to be in the vacuum state $\Omega$. Furthermore, 
since $\sigma$, $\omega$ and $T_{anomalous}$ vanish in flat background, 
 the equation \eqref{eq:quantum-one-interaction-YF} reads
\begin{equation*}
\ddot{\phi}+\frac{1}{3}\lambda\dot{\varphi}^2\phi=0\;.
\end{equation*}
The explicit expression for the retarded fundamental solution $R_V$, obtained imposing initial asymptotic condition for $t\rightarrow-\infty$, is
\begin{equation}
R_V(h)(t)
= \int h(s)\Theta(t-s)ds,
\end{equation}
where $\Theta(x):= S(x)\vartheta(x):= x\vartheta(x)$.
Hence, the classical solution of this problem enjoys 
\begin{equation}
\phi(t)=\phi_0-\frac{1}{3}\int ds\ \Theta(t-s)\lambda(s)\dot{\varphi}^2(s)\phi(s)\;,
\end{equation}
where, since $\varphi$ vanishes classically, the classical solution $\phi=\phi_0$ is chosen to be a constant.
%

We shall now pass to discuss the form of the quantum field describing the fluctuations passively induced by the matter ones. After smearing with a Gaussian shaped test function $f$, we shall evaluate both $\langle\phi(f)\rangle$ and  $\varsigma^2(f):=\langle(\phi(f)-\langle\phi(f)\rangle)^2\rangle$  on the Minkowski ground state expanding $\phi$ up to the second perturbative order.
It is thus possible to approximate his probability distribution with a Gaussian after smearing with suitable compactly supported smooth functions.

Doing this calculation we use the pull-back on the algebra $\mathcal{A}(\gamma)$ of the Minkowski vacuum state under the restriction map on any geodesic $\gamma$ of the congruence.
The obtained two-point function, parametrized with respect of the proper time, is  
\begin{equation}\label{eq:2pt-mink-rest}
\Omega_2(t,t') := \frac{1}{4\pi^2}\frac{1}{(t-t'-i\epsilon)^2} \;.
\end{equation}
This two-point function can be used to construct the $\star_\omega$ product on $\mathcal{F}(\mathbb{R})$. 

We shall now write the explicit form of $\boldsymbol{\phi}(f)$ up to the second order in the coupling parameter, which is the first non-trivial perturbation order.
Since, for the chosen congruence $\phi_0$ is constant, we have immediately from \eqref{eq:phi0e1-r} that 
\begin{equation}\label{eq:first-order-mink}
\boldsymbol{\phi}_1(f) = \int f(x_1) R_V(x_1,x_0) \lambda(x_0) \dot\varphi^2(x_0)  \phi_0\; dx_1dx_0\;.
\end{equation}
To write the explicit form of the $\boldsymbol{\phi}_2$ we need to construct 
$r(x_1,x_0)$ employing the definition \eqref{eq:def-r} valid when $x_1>x_0$ and extending it to the full space along the guidelines discussed in section \ref{sec:extension}.
We have that, up to the diagonal, it is  formed by three contributions
$$
^0r(x_1,x_0) := 
\frac{9}{\pi^4}\frac{\vartheta(x_1-x_0)}{(x_1-x_0)^8}+
\frac{6}{\pi^2} \frac{\vartheta(x_1-x_0)}{(x_1-x_0)^4} \dot\varphi(x_1)\dot\varphi(x_0)+
\vartheta(x_1-x_0)  \dot\varphi^2(x_1)\dot\varphi^2(x_0),
$$
hence it is namely the integral kernel of a distribution in $\mathcal{D}'(\mathbb{R}^2)$, which is also a functional of $\varphi\in C^\infty(\mathbb{R})$.
We need to extend it to the total diagonal $x_1=x_0$. 

The above written distribution is well defined only on $\mathcal{D}'(\mathbb{R}\backslash\{0\})$: using the techniques of scaling degree, see \cite{BFK}, the general expression for all the extensions of $^0r$ to $\mathcal{D}(\mathbb{R})$ which preserve his scaling degree at $0$ are
\begin{gather*}
r(x_1,x_0)
=-\frac{9}{\pi^4}\frac{1}{7!}\partial^8[\vartheta\ln](x_1-x_0)
+\sum_{\alpha=0}^7a_\alpha\delta^{(\alpha)}(x_1-x_0)+
\\
-\frac{1}{\pi^2}\partial^4[\vartheta\ln](x_1-x_0)\dot\varphi(x_1)
\dot\varphi(x_0)+\sum_{\alpha=0}^3b_\alpha\delta^{(\alpha)}(x_1-x_0)\dot\varphi(x_1)
\dot\varphi(x_0)+
\\
+\vartheta(x_1-x_0)\dot\varphi^2(x_1)
\dot\varphi^2(x_0)
\end{gather*}
where the derivative should be interpreted in distributional sense; the extension given by $r$ will be fixed by an appropriate choice of the renormalization constants $a_\alpha$ and $b_\alpha$. 
 
\subsection{Probability distribution for $\boldsymbol{\phi}(f)$}

We shall now estimate the probability of a focusing of the geodesic congruence in a fixed interval of time. In order to accomplish this task we shall compute the first two moments of that probability distribution $\mathbb{P}$ which describe of the expectation values of $\boldsymbol{\phi}(f)$ up to the second order in the coupling parameter. We shall then use the obtained values in order to find the Gaussian distributions which better approximate $\mathbb{P}$.
After that we may evaluate the probability of a collapse computing $\mathbb{P}(\boldsymbol{\phi}(f)\leq 0)$.

We proceed discussing the expectation value $\langle\boldsymbol{\phi}(f)\rangle:=\Omega(\boldsymbol{\phi}(f))$ of the field $\boldsymbol{\phi}(f)$ on the vacuum state $\Omega$. 
We notice in particular that 
$$
\langle\boldsymbol{\phi}_1(f)\rangle =0 
$$ 
while 
\begin{gather}\label{eq: nucleo integrale della media del campo dilatonico}
\langle\boldsymbol{\phi}_2(f)\rangle=
\phi_0\int dp\ f_R(p)\lambda(p)\left[\frac{9}{7!\pi^4}[(8\lambda^{(7)}\ast \vartheta\ln)(p)-(\lambda^{(8)}\ast S\vartheta\ln)(p)]
+\sum_{\alpha=0}^6(-1)^\alpha a_\alpha\lambda^{(\alpha)}(p)\right]
\end{gather}
where $\ast$ indicates the convolution.

%
As discussed in section \ref{sec:adiabatic-limit}, we can avoid infrared problems arising in the adiabatic limit, reformulating the problem in a suitable way. 
%
Here we notice that, the previous formula (\ref{eq: nucleo integrale della media del campo dilatonico}) 
gives meaningful results in the limit $\lambda\to 1$ for every test function $f$ provided $a_0$ is chosen to vanish.
With this choice one gets
$$
\lim_{\lambda\rightarrow 1}\langle\boldsymbol{\phi}(f)\rangle_\lambda=\phi_0 \int f  dt\;.
$$
\subsection{Approximated variance $\varsigma^2(f)$ and decay probability}
In order to obtain an expression for the variance $\varsigma^2{(f)} = {\langle\boldsymbol{\phi}^2(f)\rangle-\langle\boldsymbol{\phi}(f)\rangle^2}$ of the expectation value of the fluctuation $\phi$ smeared with a test function $f$ up to the second order in the $\lambda$ we just need to evaluate
$\langle\boldsymbol{\phi}_1(f)\star_\omega\boldsymbol{\phi}_1(f)\rangle$.
Recalling the form of $\boldsymbol{\phi}_1(f)$ obtained in \eqref{eq:first-order-mink} and the form of the two-point function restricted on the studied geodesic congruence given in \eqref{eq:2pt-mink-rest}, we get 
\begin{gather}
\varsigma^2(f) \simeq \langle\boldsymbol{\phi}_1(f)\star_\omega\boldsymbol{\phi}_1(f)\rangle
=
\nonumber
\\
=
\frac{\phi_0^2}{\pi^2}\lim_{\epsilon\rightarrow 0^+}\int dx_0dy_0\
f_R(x_0)f_R(y_0) \frac{\lambda(x_0)\lambda(y_0)}{(x_0-y_0+\imath\epsilon)^8},
\end{gather}
for all $f\in\mathcal{D}(\mathbb{R})$, which is valid up to the second order in the coupling parameter $\lambda$.
Passing to the Fourier transform\footnote{Our convention for the Fourier transform will be
\begin{equation*}
\widehat{f}(\xi):=\int d^dx\ f(x)e^{-\imath x\xi},\quad
f(x)=\frac{1}{(2\pi)^d}\int d^d\xi\ \widehat{f}(\xi)e^{\imath\xi x}.
\end{equation*}}, the previous expression can be simplified considerably, and for real smearing functions we have 
 finally
\begin{eqnarray}
\nonumber
\varsigma^2(f)
\simeq
\lim_{\epsilon_1,\epsilon_2\rightarrow 0^+}\frac{\phi_0^2}{7!\pi^2}\int_0^{+\infty} dp\ p^7 l_{\epsilon_1}(p,f)l_{\epsilon_2}(-p,f),
\end{eqnarray}
where $l_\epsilon(p,f):=\int d\xi\widehat{f}(-\xi)\widehat{\lambda}(\xi-p)/(\xi-\imath\epsilon)^2$.
%
The above obtained expression for $\varsigma^2(f)$ can be used to perform the adiabatic limit $\lambda\rightarrow 1$. In particular, in the case of real smearing functions $f$, this gives 
\begin{equation}\label{eq:varianza nel limite adiabatico}
\varsigma^2(f)
=\frac{\phi_0^2}{\pi^2}\frac{1}{7!}\int_0^{+\infty}dp \;p^3\overline{\widehat{f}(p)}\widehat{f}(p)\;.
\end{equation}
By a dilation of the test function $f$, i.e. smearing with $f_\tau(x):= f(x/\tau)/\tau$, we find $\varsigma^2(f_\tau,f_\tau)=\varsigma^2(f,f)/\tau^4$ which is consistent with the fact that the variance should vanish when the smearing function has a large support.
We shall furthermore notice that $\varsigma(f)$, at the second perturbative order, is not affected by any renormalization freedom.  

\vspace{0.5cm}
We notice that, even if we are working with distributions in $\mathcal{D}'(\mathbb{R}^2)$, the first two-moments of the distribution  of the expectation values of $\boldsymbol{\phi}(f)$ are meaningful even if $f$ is chosen to be a Schwartz function. Hence,  
smearing with a Gaussian $f(t):= \exp(-t^2)/\sqrt{\pi}$ centered at $t=0$ one gets 
$\varsigma^2(f)=2\phi_0^2/7!$.
The corresponding estimate for the probability of a collapse of the congruence of geodesics, is thus
\begin{equation}
\mathbb{P}(\phi(f)\leq 0)
=\frac{1}{\sqrt{2\pi}\varsigma(f)}\int_{-\infty}^0ds\ \exp\left(-\frac{(s-\phi_0)^2}{2\varsigma^2(f)}\right)
=N\left(-\frac{\phi_0}{\varsigma(f)},0,1\right),
\end{equation}
where $N(x,y,z)$ is the normal cumulative distribution, namely the function which computes the probability that a Gaussian of average $y$ and variance $z^2$ takes value in the interval $(-\infty,x)$.

Taking into account that, for a ``dilation'' $f\rightarrow f_\tau$ we find $\varsigma(f)\rightarrow\varsigma(f)/\tau^2$, one can compute the above written probability in terms of the dilatation parameter $\tau$:
\begin{equation}
\mathbb{P}(\phi\leq 0)
=N\left(-\tau^2,0,1\right),
\end{equation}
having absorbed a factor $\sqrt{2/7!}$ in the definition of $\tau$.

Fixing a value $\tau$ we can regard $\phi(f_\tau)$ as a random process of Poisson variables on a discrete interval of time with step $\tau$ (the variance of the Gaussian smearing function): the probability of having a  collapse is then well approximated  by an exponential of parameter $\beta_\tau:=\mathbb{P}(\phi_\tau\leq 0)$ and average $1/\beta_\tau$.
This conclusion is very similar to the result obtained in \cite{Carlip} in the context of two-dimensional dilation gravity. Of course the obtained result could be significantly different considering a better approximation for $\varsigma(f)$ and taking into account larger moments of the probability distribution.

\section{Conclusions}

In this paper we discussed the passive influence of quantum matter fluctuations on the expansion parameter describing a congruence of timelike geodesics. This problem has classically great importance in General Relativity, essentially for the powerful applications in the characterization of spacetime singularities \cite{HawkingEllis, Wald}.
Including it in the context of Quantum Field Theory on Curved Background represents an important step for a deeper insight into problems like the formation of singularity in a semiclassical regime or semiclassical black hole evaporations.
On the mathematical point of view, the problem we dealt with was the solution of the Raychaudhuri equation, which couples quantum matter fields with a geometrical field, the expansion parameter $\theta$ of the congruence of geodesics.
In more detail, considering a congruence of timelike geodesics and a quantum massless scalar field $\varphi$, we investigated the passive influence of fluctuations of $\varphi$ on the expansion parameter $\theta$ of the geodesics.

A similar work has been presented in \cite{Carlip}, where a probability distribution for a collapse of the congruence, which correspond to a negative divergence of $\theta$, was computed by a numerical procedure based on heuristic considerations in the case of a two-dimensional dilaton model for gravity.

We have shown that a more rigorous analysis can be performed restricting the focus on timelike geodesics and employing methods and techniques proper of perturbative Algebraic Quantum Field Theory (pAQFT).
In particular, we showed that the field describing the logarithmic fluctuations of the expansion parameter can be constructed as a formal power series on the algebra of matter fields restricted on any geodesics forming the congruence.

As an application, we applied the previously described procedure to the simple case furnished by Minkowski spacetime.
Here we have constructed $\phi$ up to the second order in $\lambda$ 
for a massless minimally coupled field in the vacuum state.
We have then estimated the probability distribution of $\phi$ with the Gaussian distribution, whose first two moments equal the expectation values previously computed. 
Proceeding in this way, after smearing with a suitable test function, and considering a random process on a discrete interval of time of fixed step, it is possible to estimate the mean time to have a collapse as $1/\beta_\tau,\ \beta_\tau:=\mathbb{P}(\phi(f_\tau)<0)$, obtaining results qualitatively similar to the one of \cite{Carlip}.

\section*{Acknowledgments}
  We would like to thank T.-P. Hack and V. Moretti for helpful discussions.
  We would also like to thank the referee for suggesting us a clear and elegant treatment of the adiabatic limit.
  The work of N.P. has been supported partly by the Indam-GNFM project ``Influenza della materia quantistica sulle fluttuazioni gravitazionali'' (2013).

\end{document}